\newcommand{\bolding}{1}

\documentclass[a4paper]{report}
\usepackage[utf8]{inputenc}
\usepackage[T1]{fontenc}
\usepackage{RJournal}

\usepackage{amsmath,amssymb,amsfonts,amsthm,array,nccmath}%
\usepackage{booktabs}
\usepackage{graphicx}%
\usepackage{multirow}%
\usepackage{bm}%
\usepackage{mathrsfs}%
\usepackage[title]{appendix}%
\usepackage{xcolor}%
\usepackage{textcomp}%
\usepackage{manyfoot}%
\usepackage{microtype}%
\usepackage{booktabs}%
\usepackage{algorithm}%
\usepackage{algorithmicx}%
\usepackage{algpseudocode}%
\usepackage{listings}%
\usepackage{pxfonts}%
\usepackage{xcolor}%
\usepackage{subcaption}
\usepackage{natbib}
\usepackage[font = small]{caption}
\usepackage{cleveref}%

\begin{document}

\sectionhead{Contributed research article}
\volume{XX}
\volnumber{YY}
\year{20ZZ}
\month{AAAA}

\begin{article}
\newcommand{\T}{\mathrm{\scriptscriptstyle T} }
\newcommand{\trans}{^\top} 
\newcommand{\prm}{^{\prime}} 
\newcommand{\dprm}{^{\prime\prime}} 
\newcommand{\fim}[1]{\trans{\bX}#1\bX} 
\newcommand{\dd}[1]{\mfrac {\mathrm{d}}{\mathrm{d}#1}} 
\newcommand{\del}[2]{\frac{\partial#1}{\partial #2}} 
\newcommand{\deltwo}[1]{\mfrac{\partial^2}{\partial #1^2}} 
\newcommand\scalemath[2]{\scalebox{#1}{\mbox{\ensuremath{\displaystyle #2}}}}
\newcommand{\beginsupplement}{%
        \setcounter{equation}{0}
        \renewcommand{\theequation}{S\arabic{equation}}
        \setcounter{section}{0}
        \renewcommand{\thesection}{S\arabic{section}}
        \setcounter{table}{0}
        \renewcommand{\thetable}{S\arabic{table}}%
        \setcounter{figure}{0}
        \renewcommand{\thefigure}{S\arabic{figure}}%
     }
\newcommand{\Exp}{\mathbb{E}}
\newcommand{\Cov}{{\rm Cov}}
\newcommand{\Var}{{\rm Var}}
\newcommand{\vecc}{{\rm vec}} 
\newcommand{\vech}{{\rm vech}} 
\newcommand{\diag}{{\rm diag}} 
\newcommand{\tildeK}{\widetilde{K}} 
\newcommand{\nablaa}{\widetilde{\nabla}} 
\newcommand{\Partials}{\if1\bolding{\widetilde{\bm{\partial}_\bs}}\fi\if0\bolding{\widetilde{\partial}_\bs}\fi}
\newcommand{\Sigmaa}{\widetilde{\Sigma}} 
\newcommand{\given}{\, | \;} 
\newcommand{\mstar}{m^*}
\def\A{\mathscr{A}}
\def\C{\mathscr{C}}
\def\G{\if1\bolding{\bm{\mathscr{G}}}\fi\if0\bolding{\mathscr{G}}\fi}
\def\I{\mathcal{I}}
\def\calK{\if1\bolding{\bm{\mathcal{K}}}\fi\if0\bolding{\mathcal{K}}\fi}
\def\L{\if1\bolding{\bm{\mathcal{L}}}\fi\if0\bolding{\mathcal{L}}\fi}
\def\M{\mathcal{M}}
\def\N{\mathcal{N}}
\def\P{\mathcal{P}}
\def\D{\mathcal{D}}
\def\Y{\if1\bolding{\bm{\mathcal{Y}}}\fi\if0\bolding{\mathcal{Y}}\fi}
\def\Z{\if1\bolding{\bm{\mathcal{Z}}}\fi\if0\bolding{\mathcal{Z}}\fi}
\def\R1{\mbox{$\mathfrak{R}$}}
\def\R2{\mbox{$\mathfrak{R}^2$}}
\def\R3{\mbox{$\mathfrak{R}^3$}}
\def\Rd{\mbox{$\mathfrak{R}^{d}$}}
\def\Rq{\mbox{$\mathfrak{R}^{q}$}}
\newtheorem{definition}{Definition}
\newtheorem{theorem}{Theorem}
\newtheorem{result}{Result}
\newtheorem{condition}{Condition}
\newtheorem{lemma}{Lemma}
\newtheorem{corollary}{Corollary}
\newtheorem{proposition}{Proposition}
\newtheorem{assumption}{Assumption}
\newtheorem{remark}{Remark}
\def\bnabla{\if1\bolding{\bm{\nabla}}\fi\if0\bolding{\nabla}\fi}
\def\bpartial{\if1\bolding{\bm{\partial}}\fi\if0\bolding{\partial}\fi}
\def\0{\if1\bolding{\bm{0}}\fi\if0\bolding{0}\fi}
\def\1{\if1\bolding{\bm{1}}\fi\if0\bolding{1}\fi}
\def\ba{\if1\bolding{\bm{a}}\fi\if0\bolding{a}\fi}
\def\be{\if1\bolding{\bm{e}}\fi\if0\bolding{e}\fi}
\def\bu{\if1\bolding{\bm{u}}\fi\if0\bolding{u}\fi}
\def\bv{\if1\bolding{\bm{v}}\fi\if0\bolding{v}\fi}
\def\bw{\if1\bolding{\bm{w}}\fi\if0\bolding{w}\fi}
\def\bW{\if1\bolding{\bm{W}}\fi\if0\bolding{W}\fi}
\def\bx{\if1\bolding{\bm{x}}\fi\if0\bolding{x}\fi}
\def\bs{\if1\bolding{\bm{s}}\fi\if0\bolding{s}\fi}
\def\bDelta{\if1\bolding{\bm{\Delta}}\fi\if0\bolding{\Delta}\fi}
\def\bP{\if1\bolding{\bm{\mathrm{P}}}\fi\if0\bolding{P}\fi} 
\def\bmu{\if1\bolding{\bm{\mu}}\fi\if0\bolding{\mu}\fi}
\def\btheta{\if1\bolding{\bm{\theta}}\fi\if0\bolding{\theta}\fi}
\def\bbeta{\if1\bolding{\bm{\beta}}\fi\if0\bolding{\beta}\fi}
\def\bgamma{\if1\bolding{\bm{\gamma}}\fi\if0\bolding{\gamma}\fi}
\def\bL{\if1\bolding{\bm{L}}\fi\if0\bolding{L}\fi}
\def\bV{\if1\bolding{\bm{\mathrm{V}}}\fi\if0\bolding{V}\fi} 
\def\bn{\if1\bolding{\bm{n}}\fi\if0\bolding{n}\fi}
\def\m{\if1\bolding{\bm{m}}\fi\if0\bolding{m}\fi}
\def\bH{\if1\bolding{\bm{\mathrm{H}}}\fi\if0\bolding{H}\fi} 
\def\bN{\if1\bolding{\bm{\mathrm{N}}}\fi\if0\bolding{N}\fi} 
\def\bGamma{\if1\bolding{\bm{\Gamma}}\fi\if0\bolding{\Gamma}\fi}
\def\bF{\if1\bolding{\bm{F}}\fi\if0\bolding{F}\fi}
\def\bA{\if1\bolding{\bm{\mathrm{A}}}\fi\if0\bolding{A}\fi} 
\def\bB{\if1\bolding{\bm{\mathrm{B}}}\fi\if0\bolding{B}\fi} 
\def\bK{\if1\bolding{\bm{\mathrm{K}}}\fi\if0\bolding{K}\fi} 
\def\bG{\if1\bolding{\bm{\mathrm{G}}}\fi\if0\bolding{G}\fi} 
\def\bI{\if1\bolding{\bm{\mathrm{I}}}\fi\if0\bolding{I}\fi} 
\def\bX{\if1\bolding{\bm{\mathrm{X}}}\fi\if0\bolding{X}\fi} 
\def\bR{\if1\bolding{\bm{\mathrm{R}}}\fi\if0\bolding{R}\fi} 
\def\bE{\if1\bolding{\bm{\mathrm{E}}}\fi\if0\bolding{E}\fi} 
\def\bO{\if1\bolding{\bm{\mathrm{O}}}\fi\if0\bolding{O}\fi} 
\def\bU{\if1\bolding{\bm{\mathrm{U}}}\fi\if0\bolding{U}\fi}
\def\bM{\if1\bolding{\bm{\mathrm{M}}}\fi\if0\bolding{M}\fi}
\title{nimblewomble: An R package for Bayesian Wombling with \texttt{nimble}}
\author{by Aritra Halder and Sudipto Banerjee}

\maketitle

\abstract{
This exposition presents \CRANpkg{nimblewomble}, a software package to perform wombling, or boundary analysis, using the \texttt{nimble} Bayesian hierarchical modeling language in the \textbf{R} statistical computing environment. Wombling is used widely to track regions of rapid change within the spatial reference domain. Specific functions in the package implement Gaussian process models for point-referenced spatial data followed by predictive inference on rates of change over curves using line integrals. We demonstrate model based Bayesian inference using posterior distributions featuring simple analytic forms while offering uncertainty quantification over curves.
}


\section{Introduction}\label{sec:intro}
Detecting regions of rapid change is an important exercise in spatial data science as they harbor effects not easily explained by predictors incorporated into a spatial regression model for point-referenced spatial data. For example, environmental health scientists are often keen on identifying regions where exposure levels display rapid change or sharp gradients. Formal statistical detection of such regions can lead to data-driven discoveries of latent risk factors and other predictors that drive the rapid change in exposure surfaces.  
 This exercise is referred to as wombling \citep[][]{womble_differential_1951,gleyze2001wombling,banerjee2010handbook}. { Wombling results in curves that track regions of interest. Identifying such regions serves as crucial guides for interventions. For example, to determine whether natural features such as rivers or mountain ranges represent significant zones of rapid change in weather patterns; To identify significant boundaries (over geographic areas) in disease rates, such as cancer, or to track variations in access to health care across geographical regions.} Measurement scales of the spatial data usually dictate the methods required for wombling. For areal data, boundaries delineate neighboring regions \citep[see, e.g.][]{gao_spatial_2023,wu_assessing_2025}. Predictive inference is sought for \emph{smooth} curves. We evaluate spatial gradients along a curve while assessing its candidacy for a boundary \citep[see, e.g.,][]{banerjee2006bayesian, qu_boundary_2021,halder2024bayesian}, which requires specifying the smoothness of the spatial process \citep[see, e.g.][]{kent_continuity_1989,banerjee_directional_2003}.

In this software package, we are concerned with point-referenced wombling. Developing an easily accessible software that implements Bayesian wombling for use by the wider scientific community is faced with several challenges. Perhaps, the most severe being the need for two-dimensional quadrature to enable posterior inference. Our contributions here lie in the use of analytic closed forms for posteriors that require at most one-dimensional quadrature, greatly easing the computational burden and efficient Bayesian inference for hierarchical spatial models \citep[see, e.g.][]{banerjee_hierarchical_2014} via \CRANpkg{nimble} \cite[][]{de_valpine_programming_2017} within the {\bf R} \cite[][]{r_core_team_2021} statistical environment.

Several {\bf R} packages exist for point-referenced spatial modeling, with \CRANpkg{spBayes} \cite[][]{finley_spbayes_2007} and \CRANpkg{R-INLA} \cite[][]{lindgren_bayesian_2015} being more widely used. However, they do not address boundary analysis, or wombling, in any capacity. In recent years, \CRANpkg{nimble} has found increased use in Bayesian modeling applications \citep[see, e.g.][]{turek_efficient_2016,ponisio_one_2020,goldstein_comparing_2022}. Notable {\bf R}-packages that use nimble include \CRANpkg{BayesNSGP} \cite[][]{risser_bayesian_2020} and \CRANpkg{ nimbleEcology} \cite[][]{nimbleEcology_2024}. The Bayesian hierarchical 
framework in \CRANpkg{nimblewomble} is similar to \CRANpkg{spBayes}. We take advantage of the \emph{one line call and execute} feature of {\bf nimble} to develop Markov Chain Monte Carlo (MCMC) algorithms for fitting Gaussian process (GPs). This makes the underlying code for \CRANpkg{nimblewomble} easily accessible and customizable for wider use. 

We demonstrate our developments using the Mat\'ern class of covariance kernels \cite[see, e.g.,][]{abramowitz1988handbook}. They are a popular choice in the literature for GPs \citep[see, e.g.,][]{rasmussen_gaussian_2005}. They feature a fractal parameter that provides explicit control over process smoothness. Our package offers three choices for the fractal parameter allowing for flexible process specification. Gradient estimation is done on a grid. We show an example of generating an equally spaced grid. Users can also specify a grid of their choice. Wombling requires a curve, we use contours for that purpose. We demonstrate the procedure for obtaining contours using the \CRANpkg{raster} package. Alternatively, a curve of choice can also be used. Our vignettes show an example that uses the \texttt{locator()} function. The user annotates points on the interpolated surface and a smooth B\'ezier curve is generated for use. Finally, posterior samples from models in \CRANpkg{spBayes} can also be used for wombling with \CRANpkg{nimblewomble}. The same kernel needs to be used in both packages to ensure valid inference. 

{ The ensuing discussion describes the necessary methodological details in brief. We begin with a general overview of the functions within \CRANpkg{nimblewomble}. \Cref{sec:sim} contains worked out examples to demonstrate the workflow. \Cref{sec:omics} houses an application of the package to perform boundary analysis on a spatial transcriptomics dataset.}

{
\section{General Package Overview}
The \CRANpkg{nimblewomble} package is available for download on the Comprehensive R Archive Network (CRAN). 
It contains functions that are required to perform wombling. These functions are described in \Cref{tab:package_main_functions}. The main functions can broadly be classified into four categories: covariance kernels, model fitting, inference on rates of change and line integrals and graphical displays. All functions, with the exception of plotting, are scripted as \texttt{nimbleFunction}s with wrapper functions that are callable through {\bf R}. This enables fast execution using their compiled \texttt{C++} counterparts. We generate spatial graphics using \CRANpkg{ggplot2} \cite[][]{wickham2011ggplot2} and \CRANpkg{MBA} \cite[][]{MBA}. Other internal helper functions within the package serve specific computational purposes, for example, the incomplete Gamma integral is computed by \texttt{gamma\_int}. They are primarily for internal use and are described under the internal type (row subsections) within \Cref{tab:package_main_functions}.  

\begin{table}[t]
    \centering
    \resizebox{0.95\linewidth}{!}{
         \begin{tabular}{l|@{\extracolsep{1pt}}l|l|l@{}}
          \hline
          \hline
          Type & Function & Purpose  & Description\\\hline
          \multirow{10}{*}{Main} &  \multicolumn{1}{p{3cm}|}{\texttt{materncov1} \texttt{materncov2} \texttt{gaussian}} & covariance kernel & \multicolumn{1}{p{6cm}}{Mat\'ern covariance  with $\nu = \frac{3}{2}, \frac{5}{2}$ and $\infty$ (squared exponential kernel)}\\\cline{2-4}
          &\multicolumn{1}{p{3cm}|}{\texttt{gp\_fit}}  & model fitting & \multicolumn{1}{p{6cm}}{Fits a Gaussian Process with non-informative priors. Produces posterior samples for $\btheta$}\\
          &\multicolumn{1}{p{3cm}|}{\texttt{zbeta\_samples}}  & model fitting & \multicolumn{1}{p{6cm}}{Posterior samples for $Z(\bs)$ and $\bbeta$}\\\cline{2-4}
          &\multicolumn{1}{p{3cm}|}{\texttt{sprates}}  & inference & \multicolumn{1}{p{6cm}}{Posterior samples for $\bpartial Z(\bs)$ and $\widetilde{\bpartial}^2Z(\bs)$}\\
          &\multicolumn{1}{p{3cm}|}{\texttt{spwombling}} & inference & \multicolumn{1}{p{6cm}}{Posterior samples for $\bGamma(C)$}\\\cline{2-4}
          &\multicolumn{1}{p{3cm}|}{\texttt{sp\_ggplot}}  & plotting & \multicolumn{1}{p{6cm}}{Interpolated spatial surface plots}\\
          \hline
          \multirow{20}{*}{Internal} & \multicolumn{1}{p{3cm}|}{\texttt{significance}} & -- & \multicolumn{1}{p{6cm}}{Determines significance for posterior estimates}\\\cline{2-4} 
          & \multicolumn{1}{p{3cm}|}{\texttt{pnorm\_nimble}} & -- & \multicolumn{1}{p{6cm}}{Computes the Cumulative Distribution Function (CDF) for the standard Gaussian probability distribution}\\\cline{2-4}
          & \multicolumn{1}{p{3cm}|}{\texttt{gamma\_int}} & -- & \multicolumn{1}{p{6cm}}{Incomplete Gamma Function}\\\cline{2-4}
          & \multicolumn{1}{p{3cm}|}{\texttt{gamma1.mcov1} \texttt{gamma1n2.gauss} \texttt{gamma1n2.mcov2}} & -- & 	\multicolumn{1}{p{6cm}}{Cross-covariance terms for the posterior distribution of wombling measures}\\\cline{2-4}
          & \multicolumn{1}{p{3.5cm}|}{\texttt{gradients\_matern1} \texttt{curvatures\_matern2} \texttt{curvatures\_gaussian}} & -- & 	\multicolumn{1}{p{6cm}}{Posterior samples of rates of change (gradients and curvatures) for various kernels}\\\cline{2-4}
          & \multicolumn{1}{p{3.5cm}|}{\texttt{wombling\_matern1} \texttt{wombling\_matern2} \texttt{wombling\_gaussian}} & -- & 	\multicolumn{1}{p{6cm}}{Posterior samples for wombling measures for various kernels}\\\cline{2-4}
          & \multicolumn{1}{p{3.5cm}|}{\texttt{zbeta\_matern1} \texttt{zbeta\_matern2} \texttt{zbeta\_gaussian}} & -- & 	\multicolumn{1}{p{6cm}}{Posterior samples of spatial effects and intercept for various kernels}\\\cline{2-4}
          & \multicolumn{1}{p{3cm}|}{\texttt{zXbeta}} & -- & \multicolumn{1}{p{6cm}}{Posterior samples of spatial effects and intercept in the presence of covariates}\\\cline{2-4} 
          \hline
        \end{tabular}
    }
    \caption{Summary of functions that are required for wombling using \texttt{nimblewomble}.}
    \label{tab:package_main_functions}
\end{table}
}


\section{Spatial Processes for Rates of Change}\label{sec:spatial_rates_of_change}

We consider $\{Y(\bs): \bs\in\mathscr{S}\subseteq\mathfrak{R}^2\}$ to be a univariate weakly stationary random field with zero mean and a positive definite covariance $K(\bs, \bs')= \Cov(Y(\bs), Y(\bs'))$ for locations $\bs, \bs'\in \mathfrak{R}^2$. Mean square smoothness \citep[see, e.g.,][]{stein_interpolation_1999} at an arbitrary location $\bs_0$ requires $Y(\bs_0+h\bu) = Y(\bs_0)+h\bu^{\T}\bpartial Y(\bs_0) + h^2 {\bu^{\otimes 2}}^{\T}\bpartial^{\otimes 2}Y(\bs_0) + o(h^3||\bu||^3)$, where, $\bu=(u_1,u_2)^{\T}\in \mathfrak{R}^2$ is an arbitrary vector of directions, $\bpartial$ is the gradient operator, $\bpartial Y(\bs_0) = \left(\frac{\partial}{\partial s_x}Y(\bs_0),\frac{\partial}{\partial s_y}Y(\bs_0)\right)^{\T}=(\partial_xY(\bs_0), \partial_yY(\bs_0))^{\T}$, $\otimes$ is the Kronecker vector product. Hence, $\bu^{\otimes 2}= \left(u_1^2, u_1u_2, u_2u_1, u_2^2\right)^{\T}$ and $\bpartial^{\otimes 2} = \bpartial \otimes \bpartial$. Note that $\bpartial^{\otimes 2} Y(\bs_0)$ is the vectorized Hessian. The processes $\bpartial Y(\bs_0)$ and $\bpartial^{\otimes 2}Y(\bs_0)$ govern rates of change in $Y(\bs_0)$. The \emph{gradient} or, first order rate of change, is captured by $\bpartial Y(\bs_0)$ while, \emph{curvature} is captured by $\bpartial^{\otimes 2}Y(\bs_0)$. \emph{Mean square differentiability} \citep[see][]{banerjee_smoothness_2003} of the first and second order for $Y(\bs_0)$ guarantees that $\bu^{\T}\bpartial Y(\bs_0)$ and $(\bu\otimes\bv)^{\T}\bpartial^{\otimes 2}Y(\bs_0)$ are well-defined respectively, for any set of direction vectors $\bu, \bv\in \mathfrak{R}^2$. We note that the entries of $\bpartial^{\otimes 2}Y(\bs_0)$ contain duplicates, both $\partial^2_{xy}Y(\bs_0)=\frac{\partial^2}{\partial s_x\partial s_y}Y(\bs_0)$ and $\partial^2_{yx}Y(\bs_0)$ are included. To avoid singularities that arise from duplication, we work with $\widetilde{\bpartial}^{\otimes 2}Y(\bs_0)=\left(\partial_{xx}^2Y(\bs_0),\partial_{xy}^2Y(\bs_0), \partial_{yy}^2Y(\bs_0)\right)^{\T}$ comprised of only unique derivatives.

Statistical inference is devised for the joint process, $\L^*Y(\bs) = \left(Y(\bs), \bpartial Y(\bs)^{\T}, \widetilde{\bpartial}^{\otimes 2}Y(\bs)^{\T}\right)^{\T}$. { Note that in $\mathfrak{R}^2$, $\bpartial Y(\bs)$ and $\widetilde{\bpartial}^{\otimes 2}Y(\bs)$ are $2\times 1$ and $3\times 1$ vectors respectively.} Validity of the inference is considered at length in 
\cite{banerjee_directional_2003,halder2024bayesian}. The process $\L^*Y(\bs)$ is also weakly stationary with a cross-covariance matrix { of order $6\times 6$ given by}, 

\begin{equation}\label{eq:cross-cov}
    \bV_{\L^*}(\bDelta) =\left(\begin{array}{ccc}
        K(\bDelta) &  \bpartial K(\bDelta)^{\T} & \widetilde{\bpartial}^2 K(\bDelta)^{\T} \\
         -\bpartial K(\bDelta) & -\bpartial^2K(\bDelta) & -\widetilde{\bpartial}^3K(\bDelta)^{\T}\\
         \widetilde{\bpartial}^{2}K(\bDelta) & \widetilde{\bpartial}^3K(\bDelta) & \widetilde{\bpartial}^4K(\bDelta)
    \end{array}\right), 
\end{equation}
where $\bDelta = \bs-\bs'$, $\bpartial K(\bDelta)$ is a $2\times 1$ vector of gradients, $\widetilde{\bpartial}^2K(\bDelta)$ is a $3\times 1$ vector of unique curvatures, $\widetilde{\bpartial}^3K(\bDelta)$ is a $3\times 2$ matrix of third derivatives, $\bpartial^2K(\bDelta)$ is the $2\times 2$ Hessian and $\widetilde{\bpartial}^4K(\bDelta)$ is a $3\times 3$ matrix of fourth order derivatives. Evidently, for $\bV_{\L^*}(\bDelta)$ in \cref{eq:cross-cov} to be valid all entries need to be well-defined.

Let $Y(\bs)\sim GP(0, K(\cdot;\btheta))$ denote a Gaussian process (GP) where $K(\bDelta;\btheta)= \Cov(Y(\bs), Y(\bs'))$ with process parameters $\btheta = \{\sigma^2, \phi\}$. We will denote $K(\bDelta;\btheta)=K(\bDelta)$ to ease notation. The covariance function satisfies $\sum_{i=1}^{N}\sum_{j=1}^{N}a_ia_jK(\bDelta_{ij})>0$ for any collection of coordinates $\{\bs_i:i = 1,\ldots,N\}$. Under isotropy we have $K(\bDelta)=\widetilde{K}(||\bDelta||)$. Let $\Y = \left(y(\bs_1), \ldots, y(\bs_N)\right)^{\T}$ be the observed realization over $\mathscr{S}$, $\Sigma_\Y$ be the $N\times N$ covariance matrix with entries $K(\bs_i,\bs_j)$, $i, j = 1, \ldots, N$ and $\bs_0$ an arbitrary location of interest. The joint distribution is as follows:

\begin{equation}\label{eq:joint}
   \Y, \bpartial Y(\bs_0), \widetilde{\bpartial}^2Y(\bs_0)\given \btheta\sim  \N_{N+5}\left(\0_{N+5},\left(\begin{array}{ccc}
         \Sigma_\Y & \bK_1 & \bK_2\\
         -\bK_1^{\T} & -\bpartial^2 K(\0) & -\widetilde{\bpartial}^3K(\0)^{\T}\\
         \bK_2^{\T} & \widetilde{\bpartial}^3K(\0) & \widetilde{\bpartial}^4K(\0)
    \end{array}\right)\right),
\end{equation}
where, $\N_d$ denotes the $d$-variate Gaussian distribution, $\bK_1 = \left(\bpartial K(\delta_{10})^{\T},\ldots,\bpartial K(\delta_{N0})^{\T}\right)^{\T}$, $\bK_2 = \left(\widetilde{\bpartial}^2 K(\delta_{10})^{\T},\ldots,\widetilde{\bpartial}^2 K(\delta_{N0})^{\T}\right)^{\T}$ and $\delta_{i0} =\bs_i-\bs_0$, $i = 1, \ldots, N$. The resulting posterior predictive distribution for rates of change at $\bs_0$ is  $P(\bpartial Y(\bs_0),\widetilde{\bpartial}^2Y(\bs_0)\given \Y) = \int P(\bpartial Y(\bs_0),\widetilde{\bpartial}^2Y(\bs_0)\given \Y,\btheta)\; P(\btheta\given \Y)\;d\btheta$. Posterior sampling proceeds in a one-for-one fashion corresponding to posterior samples of $\btheta$. From \cref{eq:joint} the resulting full conditional distribution is obtained as follows:
\begin{equation}\label{eq:posterior_gradient}
    \left(\begin{array}{c}
         \bpartial Y(\bs_0) \\
         \widetilde{\bpartial}^2Y(\bs_0) 
    \end{array}\right)\given \Y\sim\N_5\left(-\left(\begin{array}{c}
         \bK_1 \\
         \bK_2 
    \end{array}\right)^{\T}\Sigma_\Y^{-1}\Y, \left(\begin{array}{cc}
         -\bpartial^2 K(\0) & -\widetilde{\bpartial}^3K(\0)^{\T} \\
         \widetilde{\bpartial}^3K(\0) & \widetilde{\bpartial}^4K(\0) 
    \end{array}\right)-\left(\begin{array}{c}
         \bK_1 \\
         \bK_2 
    \end{array}\right)^{\T}\Sigma_\Y^{-1}\left(\begin{array}{c}
         \bK_1 \\
         \bK_2 
    \end{array}\right)\right).
\end{equation}
We use the Mat\'ern class of kernels, $K(||\bDelta||,\btheta) = \sigma^2\Gamma(\nu)^{-1}2^{1-\nu}\left(\sqrt{2\nu}\phi||\bDelta||\right)^\nu K_\nu\left(\sqrt{2\nu}\phi||\bDelta||\right)$, where $K_\nu(\cdot)$ is the modified Bessel function of the second kind \cite[][]{abramowitz1988handbook} featuring a fractal parameter $\nu$ that controls process smoothness, a spatial range parameter $\phi$ and an overall variance parameter $\sigma^2$.



\section{Spatial Wombling}\label{sec:spatial_wombling}
The wombling exercise seeks posterior predictive inference on line integrals
\begin{equation}\label{eq:wombling_measures}
    \bGamma(C) = \left(\int_C \bu^{\T}\bpartial Y(\bs)\;d\bs, \int_C{\bu^{\otimes 2}}^{\T}\bpartial^{\otimes 2} Y(\bs)\;d\bs\right)^{\T},
\end{equation}
where $C$ is a curve of interest to the investigator. Average wombling measures are defined as $\overline{\bGamma}(C) = \bGamma(C)/\ell(C)$, where $\ell$ is the arc-length measure. For closed curves we replace $\int$ with $\oint$ in \cref{eq:wombling_measures}. The choice of direction is crucial when measuring rates of change. The curve $C$ typically tracks a region of rapid change in the reference domain and hence, the direction normal to $C$ is naturally of interest. We denote the normal to $C$ at $\bs$ by $\bn(\bs)$ and set $\bu = \bn(\bs)$ in the line integrals of \cref{eq:wombling_measures}. The curve $C$ is deemed to be a \emph{wombling boundary} if any entry of $\bGamma(C)$ is large. Focusing on the choices for $C$, not all curves ensure the existence of $\bn(\bs)$ at every $\bs$. Parametric smooth curves offer some respite in that regard. We work with $C = \{\bs(t) = (s_1(t), s_2(t)):t \in\I \subset \mathfrak{R}\}$. As $t$ varies over $\I$, $\bs(t)$ traces out $C$. We assume $||\bs'(t)||\ne0$ which ensures $\bn(\bs)=||\bs'(t)||^{-1}\left(s_2'(t), -s_1'(t)\right)^{\T}$ is well-defined. 

The arc-length, $\ell(C) = \int_\I||\bs'(t)||\;dt$. For parametric curves $\bGamma(C)$ can be expressed as, $\bGamma(C) = \left(\int_\I \bn(\bs(t))^{\T}\bpartial Y(\bs(t))||\bs'(t)||\;dt, \int_\I{\bn(\bs(t))^{\otimes 2}}^{\T}\bpartial^{\otimes 2} Y(\bs(t))||\bs'(t)||\;dt\right)^{\T}$. Let $\I = [0,t^*]$, and the curve traced out over $\I$ be denoted as $C_{t^*}$. Statistical inference for $\bGamma(C_{t^*})$ follows from $\bpartial Y(\bs)$ and $\widetilde{\bpartial}^2Y(\bs)$ being GPs, as seen in \cref{eq:posterior_gradient}, $\bGamma(C_{t^*})\sim\N_2(\0_2, \bK_\bGamma(t^*,t^*))$, where $\bK_\bGamma(t^*,t^*)$ is a $2\times 2$ matrix with entries 
\begin{equation}\label{eq:variance_wm}
    k_{ij}(t^*,t^*)= (-1)^i\int_0^{t^*}\int_0^{t^*}\ba_i(t_1)^{\T}\;\bpartial^{i+j}K(\bDelta(t_1,t_2))\;\ba_j(t_2)\;||\bs'(t_1)||\;||\bs'(t_2)||\;dt_1\;dt_2,
\end{equation}
where $\ba_1(t)=\bn(\bs(t))$, $\ba_2(t)=\mathcal{E}_2\;\bn(\bs(t))^{\otimes 2}$, with $\mathcal{E}_2= \left(\begin{smallmatrix}
        1 & & & \\ & 1 & 1 & \\ & & & 1
    \end{smallmatrix}\right)$ being an elimination matrix and $\bDelta(t_1,t_2)=\bs_2(t)-\bs_1(t)$ for $i,j = 1,2$. Predictive inference on $\bGamma(C)$ requires
\begin{equation}\label{eq:joint-wm}
   \Y, \bGamma(C_{t^*})\given \btheta\sim  \N_{N+2}\left(\0_{N+2},\left(\begin{array}{cc}
         \Sigma_\Y & \G_\bGamma(t^*)^{\T} \\
         \G_\bGamma(t^*) & \bK_\bGamma(t^*, t^*)
    \end{array}\right)\right),
\end{equation}
where $\G_\bGamma(t^*)^{\T}=\left(\begin{smallmatrix}\bgamma_1(t^*)^{\T} \\ \vdots \\\bgamma_N(t^*)^{\T}\end{smallmatrix}\right)$ is an $N\times 2$ matrix with entries 
\begin{equation}\label{eq:cross-cov-wm}
    \gamma_i(t^*)=\left(\int_0^{t^*}\bn(\bs(t))^{\T} \bpartial K(\bDelta_j(t))\;||\bs'(t)||\;dt, \int_0^{t^*}{\bn(\bs(t))^{\otimes 2}}^{\T} \bpartial^{\otimes 2} K(\bDelta_j(t))\;||\bs'(t)||\;dt\right)^{\T},
\end{equation}
where $\bDelta_j(t) =  \bs(t) - \bs_j$. Posterior predictive inference proceeds one-for-one (similar to \cref{eq:posterior_gradient}) using $\bGamma(C_{t^*})\given \Y\sim \N_2\left(-\G_\bGamma(t^*)\Sigma_\Y^{-1}\Y, \bK_\bGamma(t^*, t^*) - \G_\bGamma(t^*)\Sigma_\Y^{-1}\G_\bGamma(t^*)^{\T}\right)$.

In practice, modern computing environments store curves as a set of points. As a result, it suffices to demonstrate wombling for rectilinear approximations to smooth curves where predictive inference is performed iteratively on segments. We show the inference for one generic segment. Let $C =\{\bs(t_0), \bs(t_1), \ldots, \bs(t_{n_p})\}$, then the $i$-th segment, $C_{t_i}=\{\bs(t) = \bs(t_{i-1})+t\bu_i:t\in[0,t_i]\}$, $t_i = ||\bs(t_i)-\bs(t_{i-1})||$ and $\bu_i = t_i^{-1}(\bs(t_i)-\bs(t_{i-1}))$. Clearly, $||\bu_i||=1$, $||\bs'(t)||=1$ and the normal to $C_{t_i}$ is $\bu_i^{\perp}=(u_{i2},-u_{i1})^{\T}$. For predictive inference on $\bGamma(C_{t_i})$, note that we need $\bDelta(t_1, t_2)=(t_2-t_1)\bu_i$ in \cref{eq:variance_wm} and $\bDelta_j(t)=\bDelta_{i-1,j}+t\bu_i=(\bs_{i-1}-\bs_j) + t\bu_{i-1}$ in \cref{eq:cross-cov-wm}.

\subsection{Wombling with Closed Forms}
We acknowledge that \cref{eq:variance_wm} requires 2-dimensional quadrature which is computationally expensive to evaluate. We work with the Mat\'ern kernel for which closed form analytic expressions exist for the entries of $\bK_\bGamma(t^*,t^*)$ improving on \cite{banerjee2006bayesian,halder2024bayesian} (see Theorems 1~\&~2 in the Supplement). { Apart from reduced computation time resulting from the reduction of a surface integral to closed-form expressions, the Mat\'ern class is uniquely suited for wombling, featuring a dedicated fractal parameter that controls process smoothness and thereby guaranteeing validity of posterior inference on derivatives and line integrals \citep[see, e.g.,][]{halder2024bayesian}}. Our \texttt{R}-package, \texttt{nimblewomble} features Mat\'ern kernels with $\nu = \frac{3}{2}, \frac{5}{2}$ and $\infty$ (squared exponential).

\begin{figure}[t]
    \centering
    \includegraphics[width=\linewidth]{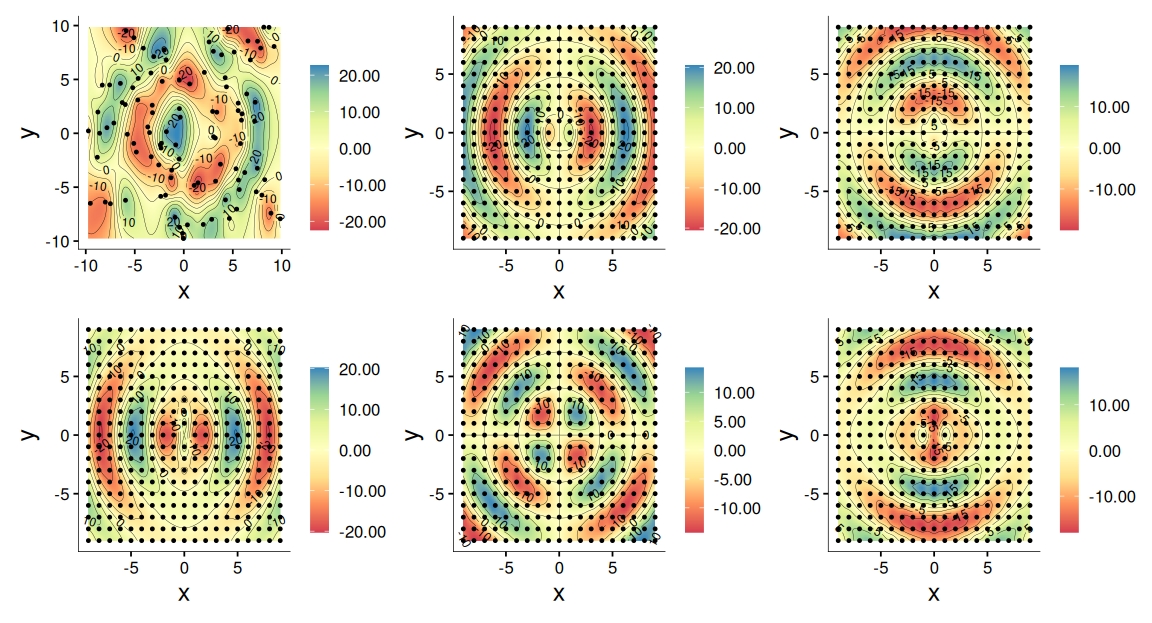}
    \caption{Patterned data used for experiments. Top row: (left) simulated process (middle) $\partial_x$ (right) $\partial_y$. Bottom row: (left) $\partial^2_{xx}$ (middle) $\partial^2_{xy}$ (right) $\partial^2_{yy}$. The grid used is overlaid on the plots.}\label{fig:true_patterns}
\end{figure}

\section{Bayesian Hierarchical Models}\label{sec:review_bhm}
    A Bayesian hierarchical model is specified as follows:
    \begin{equation}\label{eq:bhm}
        Y(\bs) = \mu(\bs, \bbeta)+Z(\bs)+\epsilon(\bs),
    \end{equation}
    where $\mu(\bs,\bbeta)=\bx(\bs)^{\T}\bbeta$, $Z(\bs)\sim GP(0,K(\cdot;\sigma^2,\phi))$ and $\epsilon(\bs)$ is a white noise process (i.e., $\epsilon(\bs_i) \overset{iid}{\sim} N(0, \tau^2)$ over any finite collection of locations). The process parameters are $\btheta =\{\sigma^2, \phi, \tau^2\}$. Predictive inference for $\L^*Z(\bs)$ evaluates $P(\bpartial Z(\bs)^{\T}, \widetilde{\bpartial}^2Z(\bs)^{\T}\given \Y)=\int P(\bpartial Z(\bs)^{\T}, \widetilde{\bpartial}^2Z(\bs)^{\T}\given \Z,\btheta)\; P(\Z\given \Y,\btheta)\; P(\btheta\given \Y)\;d\btheta\;d\Z$. Similarly for the wombling measures, $P(\bGamma_\Z(C_{t^*})\given \Y)=\int P(\bGamma_\Z(C_{t^*})\given \Z,\btheta)\;P(\Z\given \Y,\btheta)\;P(\btheta\given \Y)\;d\btheta\;d\Z$. A customary \emph{collapsed} posterior, { that is generated by marginalizing $\Z$ from the likelihood} \citep[see, e.g.][]{finley2019efficient}, for $\btheta$ is specified as follows:
    \begin{equation}\label{eq:full_posterior}
        P(\btheta\given\Y)\propto U(\phi\given a_\phi,b_\phi)\times IG(\sigma^2\given a_\sigma,b_\sigma)\times IG(\tau^2\given a_\tau,b_\tau)\times \N_N(\Y\given \bX\bbeta, \Sigma +\tau^2\bI_N),
    \end{equation}
    where $\Sigma = \sigma^2\bR_\Z(\phi)$, with $\bR_\Z(\phi)$ being the correlation matrix corresponding to $K(\cdot;\sigma^2,\phi)$, $U(\cdot\given)$ is the uniform distribution and $IG(\cdot\given)$ is the inverse-gamma distribution. Hyper-parameters are chosen such that a weakly informative prior is specified on $\btheta$. Posterior samples for $\Z$ and $\bbeta$ are generated one-for-one corresponding to posterior samples of $\btheta$ using a Gibbs sampling scheme.

    Posterior sampling in \cref{eq:full_posterior} is straightforward in \texttt{nimbleCode} as seen in the code for \texttt{gp\_model} below, which forms the core of our \texttt{gp\_fit} function (see \Cref{tab:package_main_functions}). { We use hyper-parameter settings that result in weakly informative inverse gamma priors for $\sigma^2$ and $\tau^2$.} We take advantage of the \emph{one-line call and execute} feature of \texttt{nimble} using the \texttt{buildMCMC} and \texttt{runMCMC} functions to obtain posterior samples from \cref{eq:full_posterior} thereby, fitting \cref{eq:bhm}.
    \begin{verbatim}
    #################################
    # Collapsed Metropolis-Hastings #
    #  for covariance parameters    #
    #################################
    
    gp_model <- nimbleCode({
      # Priors #
      phi ~ dunif(0, 10)
      sigma2 ~ dinvgamma(shape = 1, rate = 1)
      tau2 ~ dinvgamma(shape = 2, rate = 1)
      
      # Initialization #
      mu[1:N] <- zeros[1:N] # vector of 0s
      cov[1:N, 1:N] <- kernel(dists[1:N, 1:N], phi, sigma2, tau2)
      
      # Likelihood #
      y[1:N] ~ dmnorm(mu[1:N], cov = cov[1:N, 1:N])
    })
    \end{verbatim}
    Note that for different choices \texttt{kernel} is replaced with the corresponding kernel choice in \Cref{tab:package_main_functions}.



\section{Workflow of \texttt{nimblewomble}}\label{sec:sim}
\begin{figure}[t]
    \centering
    \includegraphics[width=\linewidth]{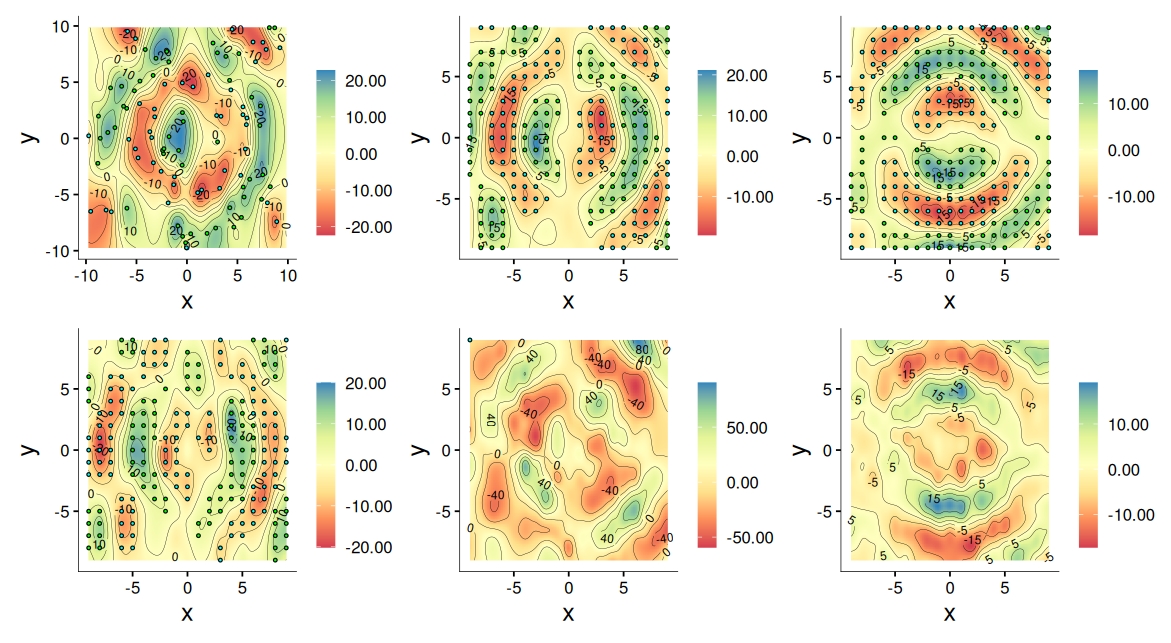}
    \caption{Estimated patterns with highlighted significant locations: positive ({\color{green} green}) negative ({\color{cyan} cyan}).}\label{fig:est_patterns}
\end{figure}
We detail the workflow for \texttt{nimblewomble} using simulated data. 
We produce data using patterns that %
yield closed form expressions for rates of change. This helps calibrate the predictive performance of \texttt{nimblewomble}. We generate $N=100$ observations over $[-10, 10] \times[-10, 10]\subseteq \mathfrak{R}^2$ arising from, $y(\bs)\sim \N_1\left(\mu_0(\bs) = 20\sin||\bs||, \tau^2\right)$. We set $\tau^2 = 1$. Here, the true values of gradients are available in closed-form. For example, $\partial_x \mu_0(\bs_G)=20\cos||\bs_G||\;s_{x,G}\;||\bs_G||^{-1}$, where $\bs_G =(s_{x,G}, s_{y,G})$ lies on a grid overlaid on the domain of reference (in this case: $[-10,10]\times[-10, 10]$). Other gradient and curvature processes are computed similarly by differentiating $\mu_0(\bs)$. Running the following code generates the simulated data and produces plots in \Cref{fig:true_patterns}.

\begin{verbatim}
set.seed(1)
# Generating Simulated Data
N = 1e2
tau = 1
coords = matrix(runif(2 * N, -10, 10), ncol = 2); colnames(coords) = c("x", "y")
y = rnorm(N, mean = 20 * sin(sqrt(coords[, 1]^2  + coords[, 2]^2)), sd = tau)

# Create equally spaced grid of points
xsplit = ysplit = seq(-10, 10, by = 1)[-c(1, 21)]
grid = as.matrix(expand.grid(xsplit, ysplit), ncol = 2)
colnames(grid) = c("x", "y")

####################################
# Process for True Rates of Change #
####################################
# Gradient along x
true_sx = round(20 * cos(sqrt(grid[,1]^2 + grid[,2]^2)) * 
                  grid[,1]/sqrt(grid[,1]^2 + grid[,2]^2), 3)
# Gradient along y
true_sy = round(20 * cos(sqrt(grid[,1]^2 + grid[,2]^2)) * 
                  grid[,2]/sqrt(grid[,1]^2 + grid[,2]^2), 3)
# Plotting
sp_ggplot(data_frame = data.frame(coords, z = y))
sp_ggplot(data_frame = data.frame(grid[-which(is.nan(true_sx)),], 
                                  z = true_sx[-which(is.nan(true_sx))]))
\end{verbatim}
We fit the model in \cref{eq:bhm} using \texttt{gp\_fit} using a Mat\'ern kernel with $\nu=\frac{5}{2}$ to the simulated data. This allows for inference on gradients and curvatures. Running the code below first generates posterior samples of $\btheta$ from \cref{eq:full_posterior} followed by posterior samples for $Z(\bs)$ and $\bbeta$ one-for-one $\btheta$. The \texttt{mc\_sp} object is a \texttt{list} comprised of (a) MCMC samples for $\btheta$ stored in \texttt{mc\_sp\$mcmc} and (b) the estimates: median and 95\% confidence intervals (CIs) stored in \texttt{mc\_sp\$estimates}. Posterior samples for $Z(\bs)$ and $\bbeta$ are obtained using \texttt{zbeta\_samples} as seen below. The \texttt{model} object contains samples for $\btheta$, $\Z$ and $\bbeta$.
\begin{figure}[t]
    \centering
    \includegraphics[width=\linewidth]{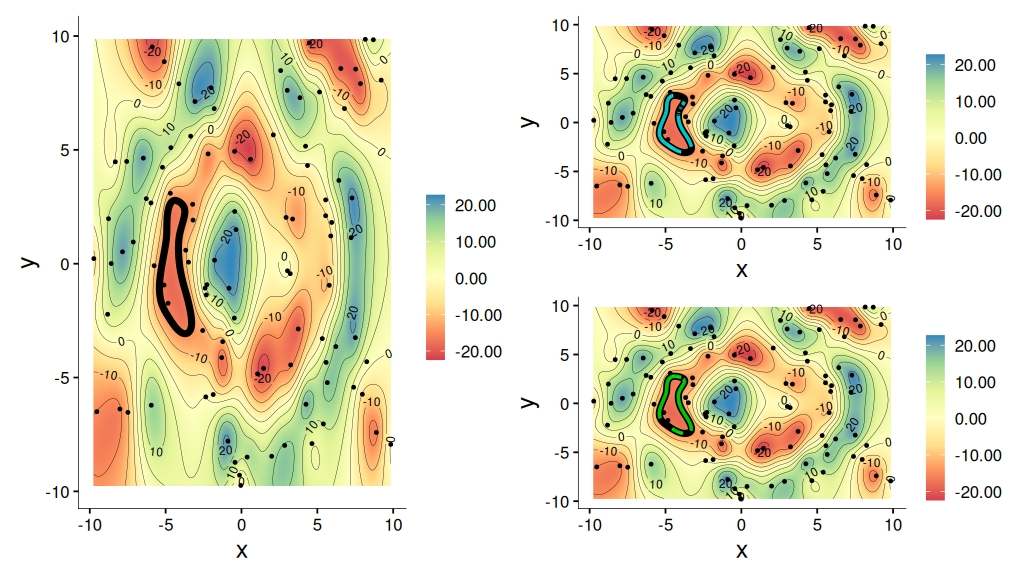}
    \caption{(Left) Curve chosen for wombling, (Right-top) gradient wombling measure for line segments, (Right-bottom) curvature wombling measure for line segments. Significant segments are highlighted: positive ({\color{green} green}) negative ({\color{cyan} cyan}).}\label{fig:wombling}
\end{figure}
\begin{verbatim}
require(nimble)
require(nimblewomble)

##########################
# Fit a Gaussian Process #
##########################
# Posterior samples for theta
mc_sp = gp_fit(coords = coords, y = y, kernel = "matern2")
# Posterior samples for Z(s) and beta
model = zbeta_samples(y = y, coords = coords,
                      model = mc_sp$mcmc,
                      kernel = "matern2")
\end{verbatim}
Next, we estimate gradients and curvatures using the posterior samples of $\phi$, $\sigma^2$ and $\Z$ using the \texttt{sprates} function. The output stored in \texttt{gradients} contains posterior samples and estimates: median and 95\% CIs for gradients and curvatures required to produce the plots in \Cref{fig:est_patterns}. Posterior sampling is done one-for-one for samples of $\phi$, $\sigma^2$ and $\Z$. 
\begin{verbatim}
###################
# Rates of Change #
###################
gradients = sprates(grid = grid,
                    coords = coords,
                    model = model,
                    kernel = "matern2")
# Plot estimated gardients along x
sp_ggplot(data_frame = data.frame(grid,
                                  z = gradients$estimate.sx[,"50%"],
                                  sig = gradients$estimate.sx$sig))    
\end{verbatim}
The wombling exercise requires a curve. The easiest choice of curves are contours. In {\bf R}, a rasterized surface using the \CRANpkg{raster} package can be used to lift contours from the interpolated surface (as seen in the plot \Cref{fig:true_patterns} top row left). The code below shows an example. The curve is shown in \Cref{fig:wombling}.
\begin{verbatim}
require(MBA)
require(raster)
# Rasterized Surface
surf <- raster(mba.surf(data.frame(cbind(coords, z = y)),
                        no.X = 300,
                        no.Y = 300,
                        h = 5,
                        m = 2,
                        extend = TRUE, sp = FALSE)$xyz.est)
# convert raster surface to contours
x = rasterToContour(surf, nlevel = 10)

x.levels <- as.numeric(as.character(x$level))
# Curve from a region of relatively low values
curves.pm.subset = subset(x, level == -15)
\end{verbatim}
Wombling is performed on this curve using the \texttt{spwombling} function. Posterior samples of $\bGamma(C)$, where $C$ is the chosen curve, are generated one-for-one $\sigma^2$, $\phi$ and $\Z$. The code below provides an example. The output is comprised of posterior samples (\texttt{wm\$wm.mcmc}) of $\bGamma(C)$ and estimates: median and 95\% CI (\texttt{wm\$estimate.wm}). It also produces the plots in \Cref{fig:wombling}. 
\begin{figure}[t]
    \centering
    \includegraphics[width=\linewidth]{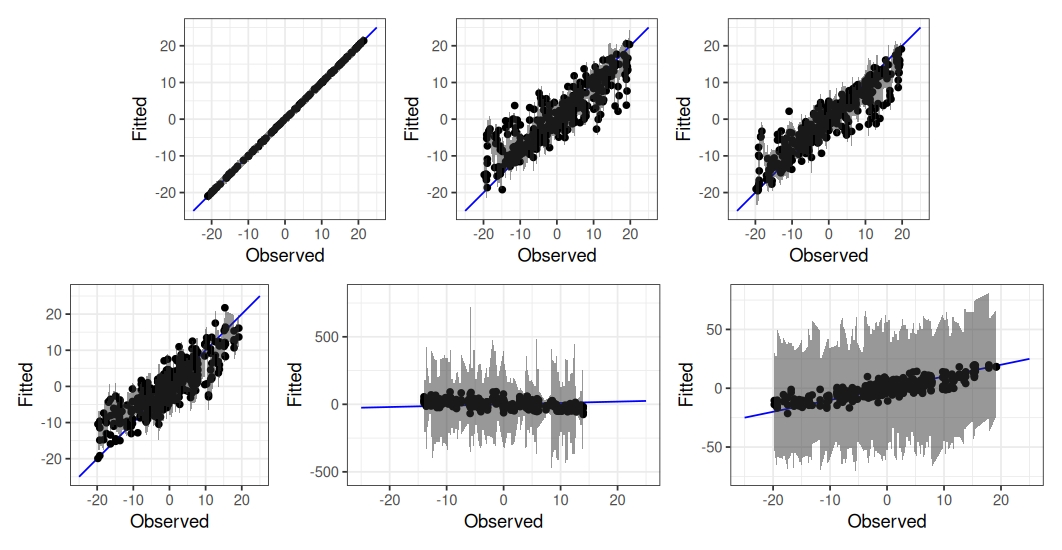}
    \caption{Diagnostics assessing quality of fit using observed vs. fitted values for (Top) (left) response (center) $\partial_x$ (right) $\partial_y$ (Bottom) (left) $\partial^2_{xx}$ (center) $\partial^2_{xy}$ (right) $\partial^2_{yy}$ with 95\% credible bands.}\label{fig:obs_v_fit}
\end{figure}
\begin{figure}[t]
            \centering
            \begin{subfigure}{.5\textwidth}
            \centering
            \hspace*{-0.6cm}    
            \includegraphics[scale = 0.6]{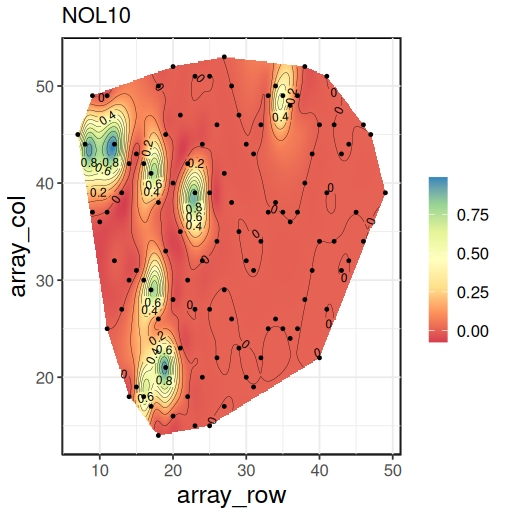}
            \caption{Low varying gene}\label{fig:lvg}
            \end{subfigure}%
            \begin{subfigure}{.5\textwidth}
            \centering
                \includegraphics[scale = 0.6]{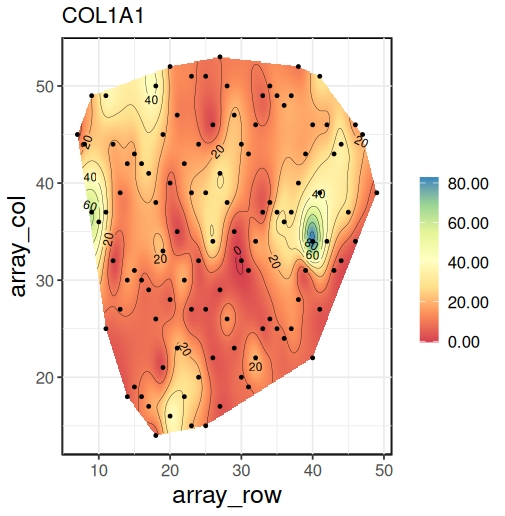}
            \caption{Spatially varying gene}\label{fig:hvg}
            \end{subfigure}
            \caption{Surfaces for gene expression of a low varying gene and spatially varying gene.}
\end{figure}
\begin{verbatim}
require(patchwork)
############
# Wombling #
############
wm = spwombling(coords = coords,
                curve = curve,
                model = model,
                kernel = "matern2")
# Total wombling measure for gradient
colSums(wm$estimate.wm.1[,-4])
# Total wombling measure for curvature
colSums(wm$estimate.wm.2[,-4])


# Color code line segments based on significance
# of gardient based wombling measure
col.pts.1 = sapply(wm$estimate.wm.1$sig, function(x){
  if(x == 1) return("green")
  else if(x == -1) return("cyan")
  else return(NA)
})
# Color code line segments based on significance
# of curvature based wombling measure
col.pts.2 = sapply(wm$estimate.wm.2$sig, function(x){
  if(x == 1) return("green")
  else if(x == -1) return("cyan")
  else return(NA)
})
######################
# Plots for Wombling #
######################
p1 = sp_ggplot(data_frame = data.frame(coords, y))
# Plot in Figure 3 (left)
p2 = p1 + geom_path(curve, mapping = aes(x, y), linewidth = 2)
# Plot in Figure 3 (top-right): gradient
p3 = p1 + geom_path(curve, mapping = aes(x, y), linewidth = 2) +
  geom_path(curve, mapping = aes(x, y), 
            colour = c(col.pts.1, NA), linewidth = 1, na.rm = TRUE)
# Plot in Figure 3 (bottom-right): curvature
p4 = p1 + geom_path(curve, mapping = aes(x, y), linewidth = 2) +
  geom_path(curve, mapping = aes(x, y), 
            colour = c(col.pts.2, NA), linewidth = 1, na.rm = TRUE)

p2 + (p3/p4) # generates Fig. 3
\end{verbatim}
We conclude the workflow with some brief comments on assessing the quality of fit. The default setting of \texttt{gp\_fit} generates $10,000$ posterior samples, with a $5,000$ burn-in. The model fit was satisfactory: $\widehat{\tau}^2 = 0.384\;(0.145,1.433)$ containing the true value of 1, $\widehat{\sigma}^2 = 344.680\;(194.292, 687.247)$ and $\widehat{\phi}=0.380\;(0.302, 0.489)$ which can be obtained from \texttt{mc\_sp\$estimates}. We achieved $\approx96\%$ coverage for the estimated rates of change and wombling measures at the line segment level. For the wombling measure, $\widehat{\bGamma(C)}=(-108.765, 154.565)^{\T}$, with 95\% CIs being $(-182.098, -36.290)$ and $(69.053, 241.664)$ respectively, containing the true values $\bGamma(C)=(-131.149, 144.010)^{\T}$. They are obtained from \texttt{wm\$estimate.wm.1} and \texttt{wm\$estimate.wm.2}. The curve $C$ forms a \emph{wombling boundary}. \Cref{fig:obs_v_fit} shows further diagnostics for model fit. 

{
\section{Further Computational Details}
We provide further details about sensitivity to prior choices, convergence and scalability that affect computational aspects that may arise in practical applications for our package.

\noindent \textbf{Alternative prior choices and sensitivity}
Although there are numerous precedents for our prior choices \citep[see, for e.g.,][]{banerjee2006bayesian,loro_bayesian_2023,halder2024bayesian}, alternative priors are available for use to the interested investigator \citep[see, for e.g.,][]{gelman2006prior}. In reference to the code snippet and discussion in \Cref{sec:review_bhm} there are several alternative choices available for $\sigma$, for example, the half Cauchy, the uniform and the folded-noncentral-$t$ distributions. Within \texttt{NIMBLE}, the uniform and half Cauchy distributions are easily implemented (\texttt{sigma $\sim$ dcauchy(0, 1) T(0, )}, where the \texttt{T(0,)} truncates the Cauchy) while, the folded-noncentral-$t$ requires some effort (using \texttt{dt\_nonstandard()}). The resulting posterior is robust to these alternative choices. The algorithm for posterior sampling under these complex priors is \emph{automatically determined} by \texttt{NIMBLE} once the model is specified. Under such complex non-conjugate priors a Metropolis Hastings algorithm is usually chosen by \texttt{NIMBLE}, which is shown in the console when fitting the model. The inference for gradients remains unaffected since it is done one-for-one using the posterior samples for $\sigma^2$, $\tau^2$ and $\phi$.  

\noindent \textbf{Convergence diagnostics} The posterior MCMC samples are stored as an \texttt{mcmc} object that communicates well the \texttt{R}-package \CRANpkg{coda}. Running \texttt{traceplot(mc\_sp)} produces trace plots to assess convergence. The \texttt{gp\_fit()} function also contains a \texttt{nchains} argument which defaults to 1, increasing this to $\geq2$ allows for computing $\widehat{R}$ using the \texttt{gelman.diag()} function for further assessment. We checked posterior convergence using trace plots which showed satisfactory evidence for convergence.

\noindent \textbf{Scalability} Addressing scalability for the package we repeated the model fit under increasing sample sizes. The model fit takes: 18.37 secs. for $N=100$, 2.52 mins. for $N=500$ and 16.39 mins. for $N=10^3$. Setting $N=100$, for a grid of size 400, generating posterior samples for gradients takes 10.02 mins. and for a curve containing 257 points, posterior samples for wombling measures takes 32.12 mins. Using a coarser representation of grids and curves (less points) improves the run-time when working with larger $N$. These computations were performed on a 12$^{\rm th}$ Gen. Intel \textregistered Core\texttrademark  i9-12.9K processor with 24 cores running Ubuntu OS and 96GB of RAM.
}

\section{\texttt{nimblewomble} in Action: Spatial Omics}\label{sec:omics}
\begin{figure}[t]
    \centering
    \includegraphics[width=\linewidth]{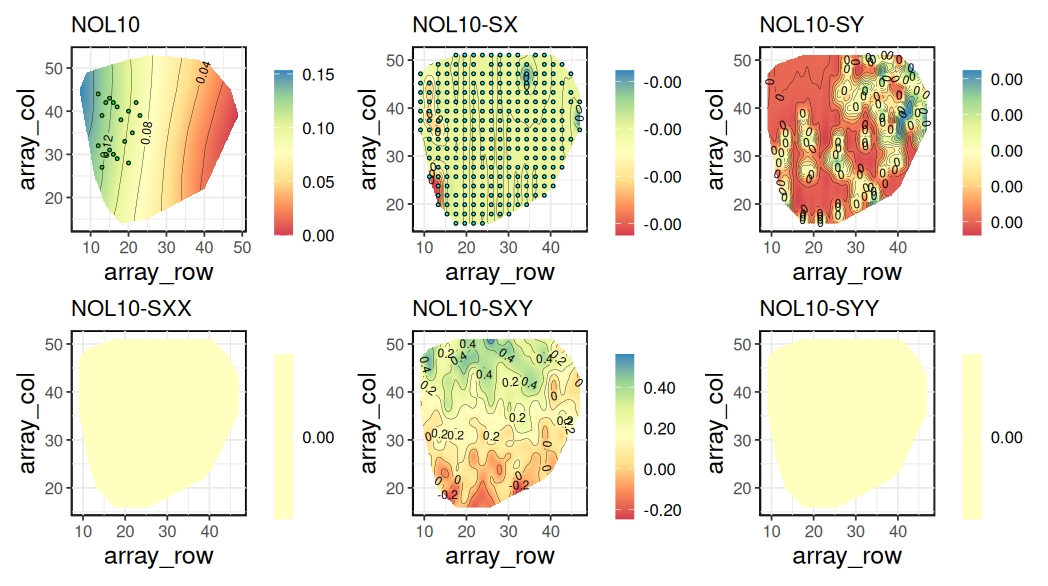}
    \includegraphics[width=\linewidth]{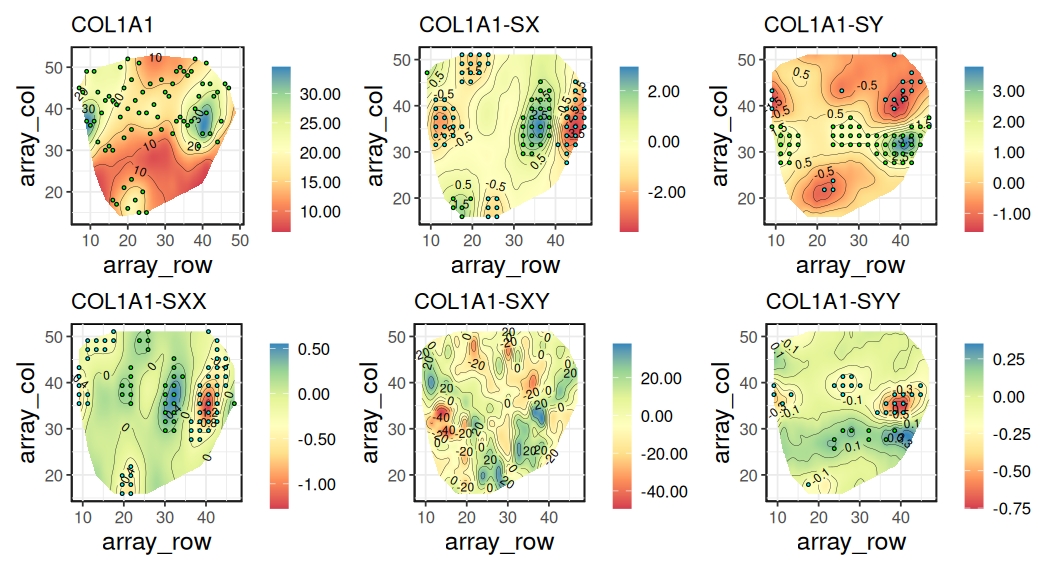}
    \caption{Plots comparing gradients for the two genes. First two rows are for the low varying gene. Bottom two rows are for the high varying gene. Significant grid locations are highlighted.}
    \label{fig:rates}
\end{figure}
We demonstrate the workflow of \texttt{nimblewomble} on a spatial omics dataset which is also supplied with the package. The data is an abridged version of what can be found in the Gene Expression Omnibus (accession number \href{https://www.ncbi.nlm.nih.gov/geo/query/acc.cgi?acc=GSE144239}{GSE144239}) \citep[see, e.g.,][]{ji_multimodal_2020}. It consists of gene expressions with tumor sampling locations for human squamous cell carcinoma, commonly known as \emph{skin cancer}. Detecting variation in gene expression is key to identifying genetic pathways specific to the cancer-type. This has led to large body of research that focuses on identifying spatially varying genes (SVGs) \citep[see, e.g.,][]{svensson_spatialde_2018,sun_statistical_2020,weber_nnsvg_2023,chen_investigating_2024}. We use rates of change to investigate differences between a SVG (\texttt{COL1A1}) and a low variance gene (\texttt{NOL10}). 

\begin{verbatim}
#################
# Load the Data #
#################
load("genes.RData")
coords = genes[, 1:2]
y = genes[, 4]; gene = "COL1A1"
N = length(y)
# Make a spatial plot of the genetic expresion
sp_ggplot(data_frame = data.frame(coords, z = y), 
          extend = FALSE, title = gene)
\end{verbatim}
The data can be loaded into the {\bf R} console by running the above code. Running \texttt{sp\_ggplot} produces an interpolated spatial plot of the raw gene expression counts as seen in \Cref{fig:hvg}. Using \texttt{genes[,3]} and running the same code produces \Cref{fig:lvg}. Comparing the ranges for the two plots the differences in expression is immediate.

We begin by fitting the GP model to individual gene expressions using \texttt{gp\_fit}. For \texttt{COL1A1}: $\widehat{\tau}^2=120.667\;(68.168, 198.105)$, $\widehat{\sigma}^2=225.652\;(78.750,529.670)$; $\widehat{\phi}=0.118\;(0.015,0.217)$, while for \texttt{NOL10}: $\widehat{\tau}^2=0.081\;(0.063, 0.017)$, $\widehat{\sigma}^2=0.676\;(0.191,6.211)$; $\widehat{\phi}=0.004\;(0.001,0.016)$. These can be accessed by running \texttt{\$estimates} on the object that stores \texttt{gp\_fit}. We use a Mat\'ern kernel with $\nu=\frac{5}{2}$.

Following up with gradient and curvature estimation using \texttt{sprates}, the resulting plots are shown in \Cref{fig:rates}. The top two rows are for \texttt{NOL10}, while the bottom two rows are for \texttt{COL1A1} which is an SVG. In each set, the first plot shows the fitted process followed by the gradients:\texttt{-SX}, \texttt{-SY} and the curvatures: \texttt{-SXX}, \texttt{-SXY} and \texttt{-SYY}. Comparing the magnitude of corresponding gradients and curvatures we see more significant grid locations show up for rates of change for the SVG as compared to the \texttt{NOL10}. This provides a more detailed picture of the manifestation of variation in expression in the SVG rather than comparing estimates of overall variance.

Finally, we pick a curve within the expression surface of \texttt{COL1A1} that tracks a region of high expression (see \Cref{fig:hvg}) and performed wombling using \texttt{spwombling}. The results for the gradient based wombling measure and the curvature based wombling measure are shown in \Cref{fig:grad} and \Cref{fig:curv} respectively; $\widehat{\bGamma}(C)=(18.176, -8.976)^{\T}$ with corresponding 95\% CI being $(-5.353, 55.965)$ for the gradient measure and $(-27.904, -0.073)$ for the curvature measure indicating that $C$ forms a {\em curvature boundary}. Curves like $C$ provide a deeper look into the tumor micro-environment.

\begin{figure}[t]
            \centering
            \begin{subfigure}{.5\textwidth}
            \centering
            \hspace*{-0.6cm}    
            \includegraphics[scale = 0.6]{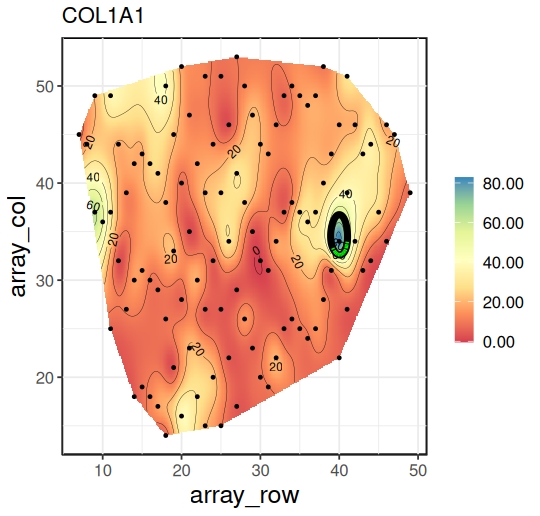}
            \caption{}\label{fig:grad}
            \end{subfigure}%
            \begin{subfigure}{.5\textwidth}
            \centering
                \includegraphics[scale = 0.6]{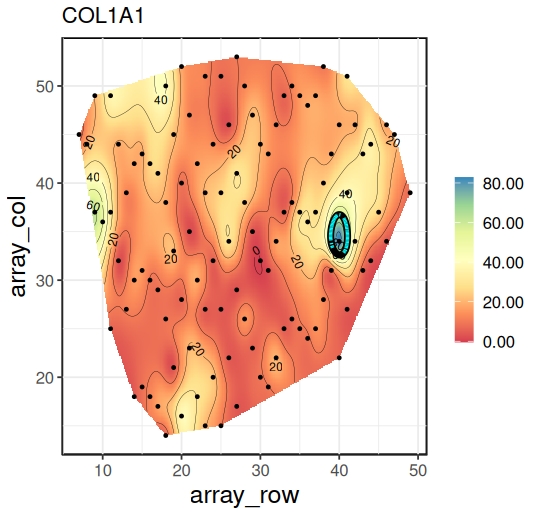}
            \caption{}\label{fig:curv}
            \end{subfigure}
            \caption{Line segment level inference for (a) gradient and (b) curvature wombling measures.}
\end{figure}

\section{Summary}
We have developed an easy-to-use software for boundary analysis or, wombling under a Bayesian framework. 
We hope that it will find use in many applications, for example usage see \cite{banerjee2006bayesian,halder2024bayesian}. The \texttt{sp\_ggplot} function also features an option to supply shape-files for more mainstream geostatistical applications. Further examples are available in the GitHub repository: \href{https://github.com/arh926/nimblewomble/}{arh926/nimblewomble/}. Boundary analysis requires the investigator to pre-select curves. Identifying such curves in context of the application often proves crucial for detecting differential behavior in the response variable. \texttt{Nimble} facilitates an accessible MCMC framework that is immensely helpful in developing the statistical inference for rates of change and boundary analysis.

Future developments can proceed along many directions. We hope to expand the software to include spatiotemporal wombling \citep[see, e.g.,][]{halder2024bayesian,halder_bayesian_2026}. Similar frameworks can be developed for generalized linear models, particularly focusing on zero-inflated models \citep[see, e.g.][]{finley_hierarchical_2011,halder_spatial_2020,halder2021spatial,halder_bayesian_2023} which provide a more realistic setting for analyzing raw gene-expression counts. We also plan to include code for inference on directional data considering Bayesian inference for the direction of maximum gradient and curvature \citep[see, e.g.,][]{wang2014projected,wang_process_2018}. Finally, developments in wombling have relied on the assumption of stationarity and isotropy. In hindsight, the assumptions yielded simplified mathematical expressions. Smoothness for nonstationary processes can be developed \citep[see, e.g.][]{paciorek2006spatial} depending on the smoothness of location specific covariances. Investigations for process smoothness can also be considered for (geometrically) anisotropic processes \citep[see, e.g.,][]{allard2016anisotropy}--we now define $K(||\Delta||)=\widetilde{K}(r)$, where $r=\sqrt{\Delta^{\T}A\Delta}$, with a positive definite matrix, $A$ corresponding to rotations and translations. The derivatives are then taken with respect to the radial profile, $r$. Finally, while this package employs MCMC for model fitting, we recognize opportunities to incorporate alternative, possibly faster, machine learning tools such as Bayesian predictive stacking \citep[see, for e.g.,][]{zhang_bayesian_2025,pan_bayesian_2025} that are implemented in the \CRANpkg{spStack} \citep[see, for e.g.,][]{pan2024spstack} package on CRAN into future editions of \CRANpkg{nimblewomble} for wombling using stacked posterior surfaces.

\section{Supplement}
\begin{theorem}\label{th:1}
         On each rectilinear segment, $C_{t^*}=\{\bs_0+t\bu:t\in[0,t^*]\}$, where $\bu$ is a unit vector and $\bu^\perp$ is its normal, the terms in the cross-covariance matrix are
         
         \begin{equation*}
             \begin{split}
                 k_{ij}(t^*,t^*)&= (-1)^i\int_0^{t^*}\int_0^{t^*}\ba_i(t_1)^{\T}\;\bpartial^{i+j}K(\bDelta(t_1,t_2))\;\ba_j(t_2)\;||\bs'(t_1)||\;||\bs'(t_2)||\;dt_1\;dt_2,\\
                 &=(-1)^i\;t^*\int_{-t^*}^{t^*}\ba_i^{\T}\;\bpartial^{i+j}K(x\bu)\;\ba_j\;dx, \quad i,j = 1,2,
             \end{split}
         \end{equation*}
         where $\ba_1(t)=\bn(\bs(t))$ is the normal to the segment, $\ba_2(t)=\mathcal{E}_2\;\bn(\bs(t))\otimes \bn(\bs(t))$, where $\mathcal{E}_2= \left(\begin{smallmatrix}
        1 & & & \\ & 1 & 1 & \\ & & & 1
    \end{smallmatrix}\right)$ is an elimination matrix and $\bDelta(t_1,t_2)=\bs_2(t)-\bs_1(t)$.
\end{theorem}
    \begin{proof}
        Considering the parametric segment, $C_{t^*}=\{\bs_0+t\bu:t\in[0,t^*]\}$, where $\bu=(u_1,u_2)^{\T}$ is a unit vector, $||\bu||=1$, $\bu^\perp$ is its normal, $\bu^{\T}\;\bu^\perp=0$, $\ba_1(t)=\ba_1=\bu^\perp$ and $\ba_2(t)=\ba_2=\mathcal{E}_2(\bu^\perp\otimes \bu^\perp)$ are free of $t$, and $\bDelta(t_1,t_2)=(t_2-t_1)\bu$. The integrand in $k_{ij}(t^*,t^*)= (-1)^i\int_0^{t^*}\int_0^{t^*}\ba_i^{\T}\;\bpartial^{i+j}K((t_2-t_1)\bu)\;\ba_j\;dt_1\;dt_2$, depends only on $(t_2-t_1)$. Making a change of variable--define $x=t_2-t_1$, $y=t_2+t_1$, implying $\frac{x+y}{2}=t_2$ and $\frac{y-x}{2}=t_1$. The Jacobian is $\frac{1}{2}$. Hence, $0\leq y+x\leq 2t^*$ and $0\leq y-x\leq 2t^*$. This implies $0 \leq y \leq 2t^*$ and $-t^*\leq x \leq t^*$. Making the substitution above reduces, $k_{ij}(t^*,t^*)= \frac{(-1)^i}{2}\int_0^{2t^*}\int_{-t^*}^{t^*}\ba_i^{\T}\;\bpartial^{i+j}K(x\bu)\;\ba_j\;dx\;dy=(-1)^i\;t^*\int_{-t^*}^{t^*}\ba_i^{\T}\;\bpartial^{i+j}K(x\bu)\;\ba_j\;dx$. Setting $i=j=1$, produces the scenario in \cite{banerjee2006bayesian}. 
    \end{proof}

    \begin{theorem}\label{th:2}
        For the Mat\'ern kernel, with $\nu = 3/2$, the variance of the gradient based wombling measure is  $k_{11}(t^*,t^*)=2\sqrt{3}\sigma^2\phi t^*G\left(1,\sqrt{3}\phi t^*\right)$. In case $\nu=5/2$, the cross-covariance matrix for the wombling measures based on spatial gradient and curvature requires $k_{11}(t^*,t^*)=\frac{2\sqrt{5}}{3}\sigma^2\phi \;t^*\left\{G\left(1,\sqrt{5}\phi t^*\right)+G\left(2,\sqrt{5}\phi t^*\right)\right\}$, $k_{22}(t^*,t^*)=10\sqrt{5}\sigma^2\phi ^3\;t^*G(1,\sqrt{5}\phi t^*)$ and $k_{21}(t^*,t^*)=-k_{12}(t^*,t^*)=0$, where $G\left(a, \frac{x}{b}\right)=\int_0^{t^*}x^{a-1}e^{-\frac{x}{b}}dx$ is the lower incomplete Gamma function with shape parameter $a$ and scale parameter $b$. For the Gaussian kernel, 

        \begin{equation*}
            \bV_{\bGamma}(t^*)=2\sigma^2\sqrt{\pi}\phi\;t^*\left\{2\Phi\left(\sqrt{2}\phi t^*\right)-1\right\}\left(\begin{smallmatrix} 1 & 0\\0&6\phi^2 \end{smallmatrix}\right),
        \end{equation*}
        where $\Phi(\cdot)$ is the cdf for the standard Gaussian probability density.
    \end{theorem}
    \begin{proof}
        Using \Cref{th:1}, if $\nu = 3/2$, then we obtain $k_{11}(t^*,t^*)=3\sigma^2\phi ^2\;t^*\int_{-t^*}^{t^*}e^{-\sqrt{3}\phi  |x|}\;dx=2\sqrt{3}\sigma^2\phi t^*\;G(1,\sqrt{3}\phi t^*)$. If $\nu=5/2$, then $k_{11}(t^*,t^*)=\frac{5}{3}\sigma^2\phi \;t^*\int_{-t^*}^{t^*}(1+\sqrt{5}\phi |x|)\;e^{-\sqrt{5}\phi  |x|}\;dx=\frac{2\sqrt{5}}{3}\sigma^2\phi \;t^*\left\{G\left(1,\sqrt{5}\phi t^*\right)+G\left(2,\sqrt{5}\phi t^*\right)\right\}$. For terms $k_{21}$ and $k_{12}$, let us consider $a_1^{\T}\bpartial^3K(x\bu)\;a_2=\frac{25}{3}\sigma^2\phi e^{-\sqrt{5}\phi  |x|}|x|\big\{\ba_1^{\T} \bA_1 \ba_2 -\sqrt{5}\phi |x|\; \ba_1^{\T}\; \bA_2 \;\ba_2\big\}$, where 
        
        \begin{equation*}
            \bA_1 = \left(\begin{smallmatrix} 3u_1 & u_2 & u_1\\u_2 & u_1 & 3u_2 \end{smallmatrix}\right), \quad \bA_2=\left(\begin{smallmatrix} u_1^3 & u_1^2\;u_2 & u_2^2\;u_1\\ u_1^2\;u_2 & u_2^2\;u_1 & u_2^3 \end{smallmatrix}\right).
        \end{equation*} 
        The first and the second term both reduce to 0. Hence, $k_{21}(t^*,t^*)=-k_{12}(t^*,t^*)=0$. Next, $\ba_2^{\T} \bpartial^4K(x\bu)\;\ba_2 = \frac{25}{3}\sigma^2\phi ^4\;e^{-\sqrt{5}\phi  |x|}\big\{\ba_2^{\T} \bA_3\; \ba_2 - \sqrt{5}\phi \; |x|\; \ba_2^{\T} \bA_4\Big(\sqrt{5}\;\phi \;|x|+1\Big) \;\ba_2\big\}$, where 
        
        \begin{equation*}
            \bA_3=\left(\begin{smallmatrix} 3 & 0 & 1\\0 & 1 & 0\\1 & 0 & 3 \end{smallmatrix}\right), \quad \bA_4(x')=\left(\begin{smallmatrix}
        6\;u_1^2-\;x'\; u_1^4 & (3-x'\; u_1^2)\;u_1\;u_2 & 1-x'\;u_1^2\;u_2^2\\ (3-x'\; u_1^2)\;u_1\;u_2 & 1-x'\;u_1^2\;u_2^2 & (3-x'\;u_2^2)\;u_1\;u_2\\1-x'\;u_1^2\;u_2^2 & (3-x'\;u_2^2)\;u_1\;u_2 & 6\;u_2^2-\;x'\; u_2^4
    \end{smallmatrix}\right).
        \end{equation*}
        After some algebra, $\ba_2^{\T}\; \bpartial^4K(x\bu)\;\ba_2=25\sigma^2\phi ^4\;e^{-\sqrt{5}\phi  |x|}\big\{1-2\sqrt{5}\;\phi \;|x|\;\bu_1^\perp \;\bu_2^\perp\;(\bu_1\;\bu_2+\bu_1^\perp\; \bu_2^\perp)\big\}$. Observe that $\bu_1^\perp = \bu_2$ and $\bu_2^\perp=-\bu_1$, substituting we get $k_{22}(t^*,t^*)=10\sqrt{5}\sigma^2\phi ^3\;t^*G(1,\sqrt{5}\phi t^*)$. 
        
        For the Gaussian kernel,
    $k_{11}(t^*,t^*)=2\sigma^2\phi^2 \;t^*\int_{-t^*}^{t^*}e^{-\phi^2  x^2}\;dx=2\sigma^2\sqrt{\pi}\phi\;t^*\left\{2\Phi\left(\sqrt{2}\phi t^*\right)-1\right\}$. 
    Note that $\ba_1^{\T}\bpartial^3K(x\bu)\;\ba_2=4\sigma^2\phi ^4\;e^{-\phi^2  x^2}x\left\{\ba_1^{\T} \bA_1 \ba_2 -2\phi^2  x^2\; \ba_1^{\T}\; \bA_2 \;\ba_2\right\}$. Again, the first and second terms equate to 0 implying $k_{21}(t^*,t^*)=-k_{12}(t^*,t^*)=0$. Next, $\ba_2^{\T} \bpartial^4K(x\bu)\;\ba_2 = 4\sigma^2\phi^4\;e^{-\phi^2 x^2}\big\{\ba_2^{\T} \bA_3\; \ba_2 - 2\phi^2 \; x^2\; \ba_2^{\T} \bA_4(2\phi^2 x^2) \;\ba_2\big\}$. After some algebra, $\ba_2^{\T}\; \bpartial^4K(x\bu)\;\ba_2=12\sigma^2\phi^4\;e^{-\phi^2  x^2}\big\{1-4\;\phi^2 \;x^2\;\bu_1^\perp\; \bu_2^\perp\;(\bu_1\;\bu_2+\bu_1^\perp\;\bu_2^\perp)\big\}$. 
    Substituting, we get $k_{22}(t^*,t^*)=12\sigma^2\sqrt{\pi}\phi^3 t^*\left\{2\Phi\left(\sqrt{2}\phi t^*\right)-1\right\}$ resulting in the required expression for $\bV_{\bGamma}(t^*)$.
    \end{proof}
    \Cref{th:2} interestingly shows that the posited cross-covariance matrix, $\bV_{\bGamma}(t^*)$, reduces to a \emph{variance-covariance matrix}. It has a simpler form for the Gaussian case when compared to Mat\'ern with $\nu =\frac{5}{2}$. 
    Finally, considering the covariance between $Z(\bs_i)$, $i=1,\ldots, N$ and the wombling measures---in the Gaussian case, closed-forms are available \citep[see, e.g.,][end of Section 3]{halder2024bayesian}. The inferential exercise of wombling does not require quadrature when using a Gaussian kernel however, only one-dimensional quadrature is required for the same when using a Mat\'ern kernel.

\bibliography{nimwomb}

\address{Aritra Halder\\
  Department of Biostatistics \& Epidemiology\\
  3215 Market Street, Philadelphia, PA 19104\\
  https://orcid.org/0000-0002-5139-3620\\
  \email{aritra.halder@drexel.edu}}

\address{Sudipto Banerjee\\
  Department of Biostatistics\\
  University of California, Los Angeles\\
  650 Charles E. Young Drive South, Los Angeles, CA 90095\\
  https://orcid.org/0000-0002-2239-208X\\
  \email{sudipto@ucla.edu}}

\end{article}

\end{document}